\documentclass[preprint,12pt]{elsarticle}

\usepackage[]{caption2}
\usepackage{graphicx}
\usepackage{subfigure}
\usepackage{amsmath}
\usepackage{amssymb}
\usepackage{marvosym}

\usepackage{algorithm}
\usepackage{algorithmic}

\numberwithin{equation}{section}
\newtheorem{theorem}{Theorem}[section]

\newtheorem{proposition}{Proposition}
\newtheorem{remark}{Remark}
\newproof{proof}{\textbf{Proof}}
\newtheorem{lemma}{Lemma}

\newtheorem{corollary}{Corollary}

\begin{document}

\begin{frontmatter}

\title{On BC-trees and BC-subtrees}

\author[Hongbo1]{Yu~Yang}
\ead{yangyugdzs@gmail.com}

\author[Deqiang]{Deqiang Wang}
\ead{dqwang@dlmu.edu.cn}

\author[Hwang]{Hua Wang}
\ead{hwang@georgiasouthern.edu}

\author[Hongbo1,Hongbo2]{Hongbo~Liu\corref{cor1}}
\ead{holiu@ucsd.edu}

\cortext[cor1]{Corresponding author: Tel: +1 858 869 9279, Fax: +1 858 822 7556}


\address[Hongbo1]{School of Information, Dalian Maritime University, Dalian 116026, China}
\address[Deqiang]{Navigation College, Dalian Maritime University, Dalian 116026, China}
\address[Hwang]{Mathematical Sciences, Georgia Southern University, Statesboro, GA 30460, USA}
\address[Hongbo2]{Institute for Neural Computation, University of California San Diego, La Jolla, CA 92093, USA}




\begin{abstract}
A BC-tree (block-cutpoint-tree) is a tree (with at least two vertices) where the distance between any two leaves is even. Motivated from the study of the ``core'' of a graph, BC-trees provide an interesting class of trees. We consider questions related to BC-trees as an effort to make modest progress towards the understanding of this concept. Constructive algorithms are provided for BC-trees with given order and number of leaves whenever possible. The concept of BC-subtrees is naturally introduced. Inspired by analogous work on trees and subtrees, we also present some extremal results and briefly discuss the ``middle part'' of a tree with respect to the number of BC-subtrees.
\end{abstract}


\end{frontmatter}


\section{Introduction}
\label{Section:Introduction}

The {\em core of a graph} $G$, $C(G)$,  was first introduced by Dulmage and Mendelsohn \cite{dulmage1958coverings,dulmage1959coverings} as part of their theory on decomposition of bipartite graphs. If $T$ is a tree and $T=C(T)$, then $T$ is a {\em BC-tree}, also known as the {\em block-cutpoint-tree} or the {\em bicolorable tree}, introduced by Harary, Plummer, and Prins \cite{har67,harary1966block}. A subtree of $T$ that is also a BC-tree is called as a {\em BC-subtree} of $T$.

It is known that a BC-tree, generally denoted by $T_{BC}$, possesses the interesting condition that the distance between any two leaves is even. In relation to a wide variety of subjects, $T_{BC}$ has a unique minimum
cover and $T_{BC}$ is the block-cutpoint-tree of some connected graphs. Two connected graphs have the same block-graph and the same cutpoint-graph if and only if they have the same block-cutpoint-tree. Barefoot \cite{barefoot2002block} showed that a tree has a BC-tree partition if and only if it does not have a perfect matching. Recently, various concepts and algorithms related to BC-tree partitions were also presented. For instance, there exists a maximum proper partial 0-1 coloring in BC-tree such that the edges colored by 0 form a maximum matching \cite{mkr06}.
It has promising potentials not only in the field of mathematics \cite{heath1999stack,AGGL2009}, but also in information science \cite{KeithPaton1971,EM2006,doerr2008directed,christou2012computing} and chemistry \cite{nakayama1983,barnard1991comparison}.

The number of subtrees of a tree, first studied in \cite{sze05}, is one of the graph invariants of trees that received much attention. Particularly, the extremal trees among certain category of trees that minimize or maximize this number have been vigorously explored. As a related concept, the number of subtrees containing at least some leaf has also been studied and shown to be related to the bound of ``acceptable residue configurations'' \cite{knu}. The analogue of such studies on BC-subtrees, however, seems to have eluded our attention.

Corresponding to the number of substructures or distances, various ``middle parts'' of a tree have received attention as an effort to understand the difference and similarity between different graph invariants. See for instance \cite{sze05,adam,barefoot,jordan}. The introduction of BC-subtrees inspires the concept of {\it BC-subtree-core} as the set of vertices contained in most BC-subtrees.

In Section \ref{sec:leafnumber}, we provide constructive algorithms that generate BC-trees with given order and number of leaves (whenever possible). In Section \ref{sec:Extremaltrees}, we consider the extremal trees with respect to the number of BC-subtrees or leaf-containing BC-subtrees. The focus is then turned to the ``middle part'' of a tree with respect to the number of BC-subtrees in Section~\ref{sec:mid}. We conclude with some comments and questions in Section \ref{Sec:Conclusion}.

\section{Constructing BC-trees with given number of leaves}
\label{sec:leafnumber}
Harary and Plummer \cite{har67} established the upper and lower bounds for the cardinality of the core of any graph $G$ with $p$ vertices that has a non-empty core. The authors further illustrated that for any integer $r$ with in the above bounds, there exists a graph $G$ with $p$ vertices and $r$ ``lines'' (edges) in its core and presented the specific constructions. Restricting our attention to BC-trees, first note that all leaves of a BC-tree belong to the same of the two independent sets of this bipartite graph (since the distance between any pair of leaves is even). Let a {\em $l$-BC-tree} denote a BC-tree with $l$ leaves, it is easy to see the following.

\begin{proposition}
There is no 2-BC-tree on $p$ vertices if $p$ is even.
\label{theorem:no-2-BC-tree}
\end{proposition}

The next observation asserts that a BC-tree cannot have exactly two internal vertices.

\begin{proposition}
There exists no $(p-2)$-BC-tree on $p$ vertices.
\label{theorem:no-p-2-BC-tree}
\end{proposition}
\begin{proof}
Assume (for contradiction) that there exists a $(p-2)$-BC-tree $T_{BC}$ on $p$ vertices. Then all other vertices are leaves adjacent to one of the two internal vertices. The distance between two such leaves with different internal neighbors is 3.
\end{proof}

Consequently, Propositions~\ref{theorem:no-2-BC-tree} and \ref{theorem:no-p-2-BC-tree} immediately imply the following.
\begin{corollary}
Let $T$ be a BC-tree of order $p\geq 3$ and $l(T)$ the number of leaves of $T$. Then $l(T) \neq p-2$ and
\begin{itemize}
\item $2\leq l(T)\leq p-1$ when p is odd;
\item $3\leq l(T)\leq p-1$ when p is even.
\end{itemize}
\label{corollary:leaf-bounds}
\end{corollary}

In fact, the conditions in Corollary~\ref{corollary:leaf-bounds} are also sufficient for the existence of a BC-tree. The constructive algorithms are provided below and illustrated in Fig.~\ref{fig:constructbc}.

\begin{theorem}
There exists an $r$-BC-tree of order $p$ if and only if $r \neq p-2$ and
\begin{itemize}
\item $2\leq r \leq p-1$ when p is odd;
\item $3\leq r \leq p-1$ when p is even.
\end{itemize}
\end{theorem}

\begin{algorithm*}
\caption{Constructing $l$-BC-trees ($l = p - 1, p - 3, p - 4,\ldots,\lceil\frac{p-1}{2}\rceil$) on $p$ vertices}
\label{Algorithm:Constructing1}
\begin{algorithmic}[1]
\STATE Initialize with $T_D$ denoting a star centered at $u_0$ with $p-1$ leaves, one of which labeled as $u_{p-1}$. Let $T_\Delta$ be a set and $T_\Delta =\{T_D\}$.

\STATE If $d_{T_D } (u_0 ) = 2$ or $d_{T_D } (u_0 ) = 3$,  go to Step \ref{Algorithm:Constructing1:Output}.

\STATE Reorganization. \label{Algorithm:Constructing1:Reorganization}

\BODY

\STATE Choose two neighbors $q,r \neq u_{p-1}$ of $u_0$.

\STATE $T_D := T_D - u_0q - u_0r + u_{p-1}q + qr$ and $T_\Delta :=T_\Delta \cup \{T_D\}$.

\ENDBODY

\STATE If $d_{T_D } (u_0 ) \leq3$, go to Step \ref{Algorithm:Constructing1:Output}. Otherwise, go to Step \ref{Algorithm:Constructing1:Reorganization}.

\STATE Output the tree set $T_\Delta$. \label{Algorithm:Constructing1:Output}
\end{algorithmic}
\end{algorithm*}

\begin{algorithm*}
\caption{Constructing $l$-BC-trees ($ l= \lceil\frac{p-1}{2}\rceil- 1,\ldots,2 $ when $p$ is odd,
$l=\lceil\frac{p-1}{2}\rceil- 1,\ldots,3$ when $p$ is even) on $p$ vertices}
\label{Algorithm:Constructing2}
\begin{algorithmic}[1]
\STATE Initialize with the last $T_D$ generated from the previous algorithm with $\lceil\frac{p-1}{2}\rceil$ leaves. Let $T_\delta$ be a set and $T_\delta =\{T_D\}$, choose a neighbor $u_1 \neq u_{p-1}$ of $u_0$. For
simplification, Let $m^*$ be the reference vertex and $m^*=u_1$.

\STATE If $d_{T_D } (u_{p-1} ) = 2$, go to Step \ref{Algorithm:Constructing2:Output}.

\STATE Reorganization. \label{Algorithm:Constructing2:Reorganization}

\BODY

\STATE Choose a pendant path $< u_{p - 1}, q, r >$ with $q \neq u_0$.

\STATE $T_D := T_D - u_{p-1}q + m^*q$, let $T_\delta :=T_\delta \cup \{T_D\}$ and $m^*=r$.

\ENDBODY

\STATE If $d_{T_D } (u_{p-1} ) = 2$, go to Step \ref{Algorithm:Constructing2:Output}. Otherwise, go to Step \ref{Algorithm:Constructing2:Reorganization}.

\STATE Output the tree set $T_\delta$. \label{Algorithm:Constructing2:Output}
\end{algorithmic}
\end{algorithm*}

\begin{figure*}[htb]
  \centering
  \subfigure[Illustration of the procedures for constructing $l$-BC-trees on $p$ vertices ($r = p - 1,p - 3,p - 4,\ldots,\lceil\frac{p-1}{2}\rceil$)]{
    \label{fig:constructbc_a} 
    \includegraphics[width=0.85\textwidth]{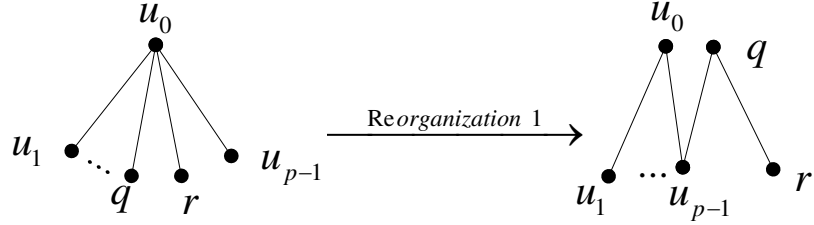}}
  \subfigure[Illustration of the procedures for constructing $l$-BC-tree ($
l=\lceil\frac{p-1}{2}\rceil- 1,\ldots,2
$)~($p$ is odd) or $l$-BC-tree ($l=\lceil\frac{p-1}{2}\rceil- 1,\ldots,3$)~($p$ is even) on $p$ vertices]{
    \label{fig:constructbc_b} 
    \includegraphics[width=0.9\textwidth]{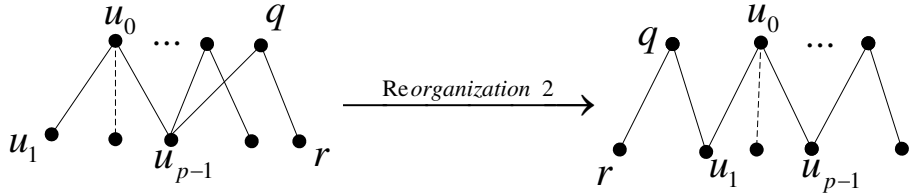}}
  \caption{Construction of $l$-BC-trees on $p$ vertices}
  \label{fig:constructbc} 
\end{figure*}

\begin{remark}
Although it is also straightforward to define such BC-trees of given order that obtain different number of leaves, it is of interests to the computational point of view to provide the above constructive algorithms.
\end{remark}

\section{Extremal trees with respect to the number of BC-subtrees}
\label{sec:Extremaltrees}

It is well known that among general trees of given order, the star maximizes the number of subtrees and the path minimizes this number (see for instance, \cite{sze05}). The explicit formulas for these numbers can also be easily obtained.
\begin{lemma}(\cite{sze05})
The path $P_n$  has ${n+1 \choose 2}$ subtrees, fewer than any other trees of $n$ vertices. The star $K_{1,n-1}$ has $2^{n-1}+n-1$  subtrees, more than any other trees of $n$ vertices.
\label{lemma:szely-wang}
 \end{lemma}

\subsection{The extremality of star and path}

Turning our attention to BC-subtrees, denote by $\eta (T)$
 (resp. $\eta_{BC} (T)$ ) the number of subtrees (resp. BC-subtrees) of $T$ and $\eta_{BC} (T,v)$ the number of BC-subtrees of $T$ containing $v$. In what follows we show the path and star, as one would expect, are also extremal with respect to the number of BC-subtrees.
\begin{theorem}The star
$K_{1,n-1}$ has $2^{n-1}-n$ BC-subtrees, more than any other trees on $n$ vertices.
\label{theorem:star-BC-subtree}
\end{theorem}
\begin{proof}
By definition, It is not difficult to obtain that
$$ \eta_{BC}(K_{1,n-1}) = 2^{n-1}-n , $$
i.e., all subtrees of $K_{1,n-1}$ except for the single-vertex and two-vertex subtrees.

For any tree $T$ with $n$ vertices, all the single-vertex or two-vertex subtrees are evidently not BC-trees. Hence
\begin{align*}
 \eta_{BC} (T) & \le \eta (T) - (|V(T)| + |E(T)|) = \eta (T) - (2n - 1) \\
  & \le (2^{n - 1}  + n - 1) - (2n-1) = \eta_{BC}(K_{1,n-1})
\end{align*}
by Lemma \ref{lemma:szely-wang}.

Furthermore, if $T$ is not the star $K_{1,n - 1}$, then there is at least one path $P$ of length 3 in $T$. This path $P$ is a non-BC-subtree of $T$ with more than two vertices and hence
\begin{equation}
\eta_{BC} (T) \le \eta (T) - (|V(T)| + |E(T)|) - 1 < \eta_{BC} (K_{1,n - 1} ).
 \nonumber
\end{equation}
\end{proof}

\begin{theorem}The number of BC-subtrees of the path $P_n$ is
\[\eta_{BC} (P_n )=\begin{cases}
{n(n-2)/ 4}&n\equiv0(mod~2),\\
(n-1)^2/ 4&n\equiv1(mod~2),
\end{cases} \]
less than that of any other $n$-vertex tree.
\label{theorem:path-BC-subtree}
\end{theorem}
\begin{proof}
Again, it is easy to obtain
$$\eta_{BC} (P_n )={n(n - 2)/4} $$
for even $n$ and
$$\eta_{BC} (P_n ) ={(n - 1)^2/4} $$
for odd $n$, i.e., the number of nontrivial subpaths of even length.

Now let $T$ be an $n$-vertex tree that is not a path with $(n \geq 4)$ (the cases for small values of $n$ is trivial).
For any $u \in V(T)$, let
 \[
E_u (T) = \{ v \in V(T)|d(u,v)\;{\rm{is}}\;{\rm{even}}\}
\]
\[
O_u (T) = \{ v \in V(T)|d(u,v)\;{\rm{is}}\;{\rm{odd}}\}
\]
where $d(u,v)$ is the distance between $u$ and $v$.

Let $|E_u (T)| = p$, $|O_u (T)| = q$, and $p+q=n$. It is easy to see that the path between any two vertices in $E_u (T)$ (resp. $O_u (T)$) is a BC-subtree of $T$. Consequently, there are ${p \choose 2}$ BC-subpaths with two endpoints in $E_u (T)$ and ${q \choose 2}$ BC-subpaths with two
endpoints in $O_u (T)$.

Since $T$ is not a path, there is at least one vertex $v$ with degree at least three. The subtree induced by $v$ and its neighbors is also a BC-subtree of $T$, hence
\[
\eta_{BC} (T) \geq {p \choose 2}+{q \choose 2}+ 1.
\]
Note that we have $p \geq 2$ and $q \geq 2$ unless $T$ is a star, in which case $\eta_{BC} (T) > \eta_{BC} (P_n)$ from Theorem~\ref{theorem:star-BC-subtree}. Assuming now $p \geq 2$ and $q\geq 2$,
\begin{itemize}
\item If $n$ is odd, we have
\[
{p \choose 2}+{q \choose 2} + 1 - \frac{{(n - 1)^2 }}{4} = \frac{{(p - q)^2  + 3}}{4} > 0 ;
\]
\item If $n$ is even, we have
\[
{p \choose 2}+{q \choose 2}+ 1 - \frac{{n(n - 2)}}{4} = \frac{{(p - q)^2 }}{4} + 1 > 0.
\]
\end{itemize}
\end{proof}

\subsection{Leaf-containing BC-subtrees}

In the rest of this section, we consider BC-subtrees containing at least one leaf, analogue of the concept of leaf-containing subtrees. We denote by $\eta_{BC}^* (T)$ the number of BC-subtrees of $T$ that contain at least a leaf of $T$.

\begin{theorem}
For any tree $T$ on $n \geq 3$ vertices, we have
$$ n-2 \le \eta_{BC}^*(T) \le 2^{n - 1}  - n $$
with equalities at the upper (lower) bound if and only if $T$ is a star (path).
\label{theorem:starpath-BC-subtree2}
\end{theorem}
\begin{proof}
The sharp upper bound simply follows from Theorem~\ref{theorem:star-BC-subtree} and the fact that every BC-subtree of a star $T$ contains some leaf of $T$.

From Theorem~\ref{theorem:path-BC-subtree} it is easy to see that
$$\eta_{BC}^ *  (P_n) = \eta_{BC} (P_n ) - \eta_{BC}  (P_{n - 2} ) = n-2 . $$

Now consider a tree $T$ of order $n$ and the two partites $X, Y$ of $T$ as a bipartite graph, denote by $n_X$ ($n_Y$) and $l_X$ ($l_Y$) the number of vertices (leaves) in $X$ and $Y$ respectively.
\begin{itemize}
\item If $l_X >0$ and $l_Y > 0$, let $w$ be a leaf in $X$, then the path connecting $w$ and any other vertex $u \in X - \{ w \}$ is a leaf-containing BC-subtree of $T$, yielding at least $(n_X - 1)$ of such subtrees. Similarly, there are at least $(n_Y - 1)$ leaf-containing BC-subtrees of $T$ formed by the paths connecting pairs of vertices in $Y$. Hence
$$ \eta_{BC}^*(T) \geq n_X - 1  + n_Y - 1 = n-2 $$
with equality if and only if $l_X = l_Y = 1$, in which case $T$ is a path.
\item Otherwise, assume without loss of generality that $l_X \geq 2$ and $l_Y = 0$. Let $n_k$ denote the number of vertices whose closest leaf is at distance $k$, i.e., $n_0=l_X$, $n_1$ is the number of internal vertices (in $Y$) adjacent to leaves. Then, by noting that all leaves are at even distance from each other, there is no edge between vertices counted by $n_k$ for any $k$. We have
$$
l_X = n_0 \geq n_1 \geq n_2 \geq n_3 \ldots \geq n_{s-1} \geq n_s = 1 \hbox{ for some $s$.}
$$
Note that $s$ is the largest distance between any internal vertex and leaf, thus $n_{s-1} > n_s$, implying that (for either odd or even $s$)
$$
n_X = n_0 + n_2 + n_4 + \ldots \geq 1 + n_1 + n_3 + n_5 + \ldots = 1 + n_Y .
$$
Consequently $n_X \geq \frac{n+1}{2}$. Considering the paths connecting a pair of vertices in $X$ with at least one of which being a leaf yields
$$ \eta_{BC}^* (T) = (n_X - l_X)l_X + {l_X \choose 2} = \left( n_X - \frac{l_X + 1}{2} \right) l_X \geq n-2, $$
with equality if and only if $l_X=2$ and $n_X = \frac{n+1}{2}$, in which case $T$ is a path of even length.
\end{itemize}
\end{proof}

\section{The BC-subtree-core of a tree}
\label{sec:mid}
In \cite{sze05}, the ``middle part'' of a tree $T$ with respect to the number of subtrees (i.e. the set of vertices contained in most subtrees) was defined as the {\em subtree-core} of $T$. It is shown that this set contains one or two adjacent vertices, a fact analogous to those of the {\em center} and {\em centroid} of a tree \cite{adam, jordan}. These facts were all proved by establishing the ``concavity'' or ``convexity'' of the corresponding counting function along any path of a tree. In terms of the number of subtrees, the following was shown in \cite{sze05}.

\begin{proposition}\label{prop:sze} (\cite{sze05})
For any three vertices $x, y, z$ of $T$ such that $xy, yz \in E(T)$, $y$ is contained in more subtrees than those containing $x$ or those containing $z$.
\end{proposition}

Inspired by the concept of subtree-core, we consider the {\em BC-subtree-core} of a tree $T$ as the set of vertices maximizing $\eta_{BC} (T,v)$ (recall that $\eta_{BC} (T,v)$ is the number of BC-subtrees of $T$ containing $v$). Unlike other ``middle parts'' of a tree, simple examination of a path shows that the BC-subtree-core is more complicated.

Let $P_n$ be a path with $V(P_n) = \{v_i|i = 1,2,\dots,n\}$, simple calculation shows
\begin{equation}
\eta_{BC} (P_n,v_i)=\begin{cases}
{\frac{i(n+1-i)}{2}-1}&i\equiv0(mod~2),\\
\lfloor\frac{i(n+1-i)-1}{2}\rfloor&i\equiv1(mod~2).
\end{cases}
\label{equ:oddeven}
\end{equation}
Consequently, one easily finds the BC-subtree-core of $P_n$ depending on $n$ modulo 4:
\begin{itemize}
   \item when $n=4k$, $\eta_{BC} (P_n,v_i)$ is maximized at $v_\frac{n}{2}$ and $v_\frac{(n+2)}{2}$ with the maximum value $\frac{(n^2+2n-8)}{8}$;
   \item when $n=4k+2$, $\eta_{BC} (P_n,v_i)$ is maximized at $v_\frac{n}{2}$ and $v_\frac{(n+2)}{2}$ with the maximum value $\lfloor\frac{(n^2+2n-4)}{8}\rfloor$ ;
   \item when $n=4k+1$, $\eta_{BC} (P_n,v_i)$ is maximized at $v_\frac{(n+1)}{2}$ with the maximum value  $\frac{(n^2+2n-3)}{8}$;
   \item when $n=4k+3$, $\eta_{BC} (P_n,v_i)$ is maximized at $v_\frac{(n-1)}{2}$, $v_\frac{(n+1)}{2}$ and $v_\frac{(n+3)}{2}$ with the maximum value $\frac{(n^2+2n-7)}{8}$.
 \end{itemize}

Noting that every BC-subtree of the star $K_{1,n-1}$ must contain the center, the following is obvious.
 \begin{proposition}
 The BC-subtree-core of the star $K_{1,n-1}(n>3)$ contains exactly the center vertex.
 \end{proposition}

In general, the BC-subtree-core of a tree needs not contain only adjacent vertices (unlike in the cases of a path or star). In Fig. \ref{fig:bccorecounterexam}, simple calculations illustrate that $\eta_{BC} (T_0,v_i)=11$ $(i=1,2,3,4)$, $\eta_{BC} (T_0,u_i)=18$ $(i=1,2)$, $\eta_{BC} (T_0,x)=\eta_{BC} (T_0,y)=21$, $\eta_{BC} (T_0,z)=19$.

\begin{figure*}[htbp]
\centering
\includegraphics[width=0.45\textwidth]{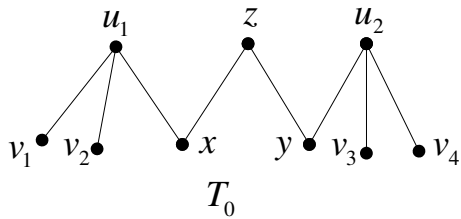}\\
\caption{The BC-subtree-core does not necessarily contain adjacent vertices.}
\label{fig:bccorecounterexam}
\end{figure*}

From observing both a path and Fig.~\ref{fig:bccorecounterexam}, it is obvious that one cannot hope for an analogue of Proposition~\ref{prop:sze}. With the special characteristics of BC-subtrees in mind, it is also interesting to ask the similar question for ``alternative'' vertices on a path. That is, for consecutive vertices $a,b,c,d,e$ on a path of a tree $T$ (i.e., $ab,bc,cd,de \in E(T)$), is $c$ necessarily contained in more BC-subtrees than $a$ and $e$?  Unfortunately, in Fig.~\ref{fig:alterbccorecounterexam}, simple calculation shows that $\eta_{BC} (T,v_i)=39$ $(i=1,2,3,4,5,6)$, $\eta_{BC} (T,u_i)=69$ $(i=1,2)$, $\eta_{BC} (T,z)=67$, $\eta_{BC} (T,x)=\eta_{BC} (T,y)=73$.
 \begin{figure*}[htbp]
\centering
\includegraphics[width=0.45\textwidth]{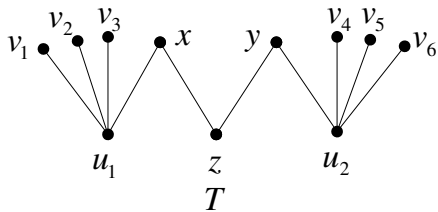}\\
\caption{No concavity exists even for alternative vertices on a path}
\label{fig:alterbccorecounterexam}
\end{figure*}

There are also many examples where the BC-subtree-core differs from the subtree-core, the center, and/or the centroid. Consider the path $T=P_3$ with vertex set $\{x,y,z\}$ where $x$ and $z$ are leaves. Then it is clear to see that the center, the centroid, and the subtree core is $\{y\}$, while the BC-subtree-core is $\{x,y,z\}$.

Another natural question is: must the subtree-core be a subset of the BC-subtree-core? Again, in Fig.~\ref{fig:bccorecounterexam}, we have $f_{T_0}(v_i)=17$ $(i=1,2,3,4)$, $f_{T_0}(u_i)=32$ $(i=1,2)$, $f_{T_0}(x)=f_{T_0}(y)=35$, $f_{T_0}(z)=36$, where $f_{T_0}(v)$ denote the number of subtrees of $T_0$ containing $v$. Therefore by definition the subtree-core of $T_0$ is $\{z\}$ while the BC-subtree-core is $\{x,y\}$.

\section{Concluding remarks}
\label{Sec:Conclusion}
Questions related to BC-trees an BC-subtrees are considered. Motivated by the study of graph core, algorithms are provided that construct BC-trees with any number of leaves when possible. Consequently, the sufficient and necessary conditions on the number of leaves for the existence of a BC-tree follow immediately. As analogous results to those on subtrees and distance-based graph invariants, the path and star are shown to be extremal with respect to the number of BC-subtrees. Although the results are similar in nature to the previous work, considering BC-subtrees turns out to be less trivial.

Regarding the number of subtrees, Yan and Yeh \cite{yan06} illustrated a linear-time algorithm to count the sum of weights of subtrees of $T$ by using the method of generating function. Considering the similar question for BC-subtrees would be interesting and natural.
Considered as a topological index or graph invariant, the number of BC-subtrees can be studied for other categories of trees such as trees with given degree sequences.

Based on the same concept, the ``middle part'' of a tree is defined and named as the BC-subtree-core, analogous to those on subtrees or distance-based invariants. A brief discussion is provided on this topic, where examples are presented showing that the BC-subtree-core behave in a rather different way than all previously known ``middle parts'' of a tree. Nevertheless, evidences suggest the following:
\begin{itemize}
\item Knowing that a BC-subtree-core does not necessarily contain adjacent vertices, must a BC-subtree-core (containing at least two vertices) contain vertices of distance at most 2?
\item Knowing that the subtree-core is not necessarily a subset of the BC-subtree-core, must they be adjacent (when they contain different vertices)?
\end{itemize}
The first question seems natural considering the special property of BC-subtrees; the second question is related to one that asks how far different ``middle parts'' can be.

Also inspired by the work on subtrees, the extremal trees with respect to the number of leaf-containing BC-subtrees are characterized. The similar concept on ``middle part'' can also be defined. It is not difficult to present similar elementary observations on the set of vertices that are contained in most leaf-containing BC-subtrees. We skip the details here.

\section{Acknowledgment}
This work is supported partly by
the Simons Foundation (Grant No. 245307),
the National Natural Science Foundation of China (Grant No. 61173 035),
and the Program for New Century Excellent Talents in University (Grant No. NCET-11-0861).


\begin{thebibliography}{10}
\expandafter\ifx\csname url\endcsname\relax
  \def\url#1{\texttt{#1}}\fi
\expandafter\ifx\csname urlprefix\endcsname\relax\def\urlprefix{URL }\fi

\bibitem{adam}
A.~\'Ad\'am, The centrality of vertices in trees, Studia Scientiarum
  Mathematicarum Hungarica 9 (1974) 285--303.

\bibitem{barefoot2002block}
C.~Barefoot, Block-cutvertex trees and block-cutvertex partitions, Discrete
  Mathematics 256~(1) (2002) 35--54.

\bibitem{barefoot}
C.~Barefoot, R.~Entringer, L.~Sz\'ekely, Extremal values for ratios of
  distances in trees, Discrete Applied Mathematics 80 (1997) 37--56.

\bibitem{barnard1991comparison}
J.~Barnard, A comparison of different approaches to markush structure handling,
  Journal of Chemical Information and Computer Sciences 31~(1) (1991) 64--68.

\bibitem{christou2012computing}
M.~Christou, M.~Crochemore, T.~Flouri, C.~S. Iliopoulos, J.~Janou{\v{s}}ek,
  B.~Melichar, S.~P. Pissis, Computing all subtree repeats in ordered trees,
  Information Processing Letters 112~(24) (2012) 958--962.

\bibitem{doerr2008directed}
B.~Doerr, E.~Happ, Directed trees: A powerful representation for sorting and
  ordering problems, in: Proceedings of IEEE Congress on Evolutionary
  Computation, IEEE, 2008, pp. 3606--3613.

\bibitem{dulmage1958coverings}
A.~L. Dulmage, N.~S. Mendelsohn, Coverings of bipartite graphs, Canadian
  Journal of Mathematics 10~(4) (1958) 516--534.

\bibitem{dulmage1959coverings}
A.~L. Dulmage, N.~S. Mendelsohn, A structure theory of bipartite graphs of
  finite exterior dimension, Transactions of the Royal Society of Canada 53~(4)
  (1959) 1--13.

\bibitem{AGGL2009}
A.~Gagarin, G.~Labelle, Two-connected graphs with prescribed three-connected
  components, Advances in Applied Mathematics 43~(1) (2009) 46--74.

\bibitem{har67}
F.~Harary, M.~Plummer, On the core of a graph, Proceedings London Mathematical
  Society 17 (1967) 249--257.

\bibitem{harary1966block}
F.~Harary, G.~Prins, The block-cutpoint-tree of a graph, Publicationes
  Mathematicae - Debrecen 13 (1966) 103--107.

\bibitem{heath1999stack}
L.~Heath, S.~Pemmaraju, Stack and queue layouts of directed acyclic graphs:
  Part \uppercase\expandafter{\romannumeral2}, SIAM Journal on Computing 28~(5)
  (1999) 1588--1626.

\bibitem{jordan}
C.~Jordan, Sur les assemblages de lignes, Journal f\"ur die Reine und
  Angewandte Mathematik 70 (1869) 185--190.

\bibitem{knu}
B.~Knudsen, Optimal multiple parsimony alignment with affine gap cost using a
  phylogenetic tree, Algorithms in Bioinformatics, Lecture Notes in Computer
  Science 2812 (2003) 433--446.

\bibitem{EM2006}
E.~Misiolek, D.~Z. Chen, Two flow network simplification algorithms,
  Information Processing Letters 97 (2006) 197--202.

\bibitem{mkr06}
V.~Mkrtchyan, On trees with a maximum proper partial 0-1 coloring containing a
  maximum matching, Discrete Mathematics 306 (2006) 456--459.

\bibitem{nakayama1983}
T.~Nakayama, Y.~Fujiwara, Computer representation of generic chemical
  structures by an extended block-cutpoint tree, Journal of Chemical
  Information and Computer Sciences 23~(2) (1983) 80--87.

\bibitem{KeithPaton1971}
K.~Paton, An algorithm for the blocks and cutnodes of a graph, Communications
  of The ACM 14 (1971) 468--475.

\bibitem{sze05}
L.~Sz{\'e}kely, H.~Wang, On subtrees of trees, Advances in Applied Mathematics
  34~(1) (2005) 138--155.

\bibitem{yan06}
W.~Yan, Y.~Yeh, Enumeration of subtrees of trees, Theoretical Computer Science
  369~(1) (2006) 256--268.

\end{thebibliography}
\end{document}